\def\longVersion{}
\documentclass[runningheads]{llncs}
  \usepackage{microtype}
  \usepackage{article}
  \usepackage{mycxx}
  \usepackage{rev}

  \usepackage{enumitem}
  \setlist[itemize]{noitemsep,topsep=0pt}

  \allowdisplaybreaks[4]

  \usepackage[numbers,sort&compress]{natbib}

  \usepackage{algorithm2e}

  \usepackage{hyperref}
  \hypersetup{pdfauthor={Akim Demaille},
              pdftitle={Derived-Term Automata of Multitape Rational Expressions},
              pdfsubject={\SvnDate \SvnRev}}
  \hypersetup{colorlinks, citecolor=blue, linkcolor=blue, urlcolor=blue}

\renewcommand{\subparagraph}[1]{\noindent\textbf{#1}~~}

\newcommand{\TD}{\mathrm{TD}}
\newcommand{\D}{\mathrm{D}}

\renewcommand{\eword}[1][]
{\ifthenelse{\equal{#1}{}}{\varepsilon}{\varepsilon_{#1}}}

\ifthenelse{\boolean{long}}
{%
  \title{Derived-Term Automata of\\
    Multitape Rational Expressions\\
    (Long version\thanks{This report is an extended
      version of the paper published in CIAA 2016 under the same name.})}
  \date{\scriptsize (\SvnDate \SvnRev)}
}
{%
  \title{Derived-Term Automata of\\
    Multitape Rational Expressions}
}
\author{Akim Demaille \email{akim@lrde.epita.fr}}

\institute{EPITA Research and Development Laboratory (LRDE)\\
  14-16, rue Voltaire, 94276 Le Kremlin-Bic\^etre, France}

\begin{document}
\mymaketitle

\begin{abstract}
  We consider (weighted) rational expressions to denote series over
  Cartesian products of monoids.  We define an operator~$\tuple$ to build
  multitape expressions such as $(a^+\tuple x + b^+\tuple y)^*$.  We
  introduce \dfn{expansions}, which generalize the concept of derivative of
  a rational expression, but relieved from the need of a free monoid.  We
  propose an algorithm based on expansions to build multitape automata from
  multitape expressions.
\end{abstract}

\ifthen{\boolean{long}}{
  \begin{description}
  \item[Changes:]~\par
    \begin{description}[noitemsep,topsep=0pt]
    \item[2016-07-25] \cref{sec:derivatives} was added, showing how to
      compute the constant term and the derivatives for the tuple operator.
      \cref{sec:related} was adapted accordingly.
    \end{description}
  \end{description}

  \mytableofcontents
}

\section{Introduction}

Automata and rational (or regular) expressions share the same expressive
power, with algorithms going from one to the other.  This fact made rational
expressions an extremely handy practical tool to specify some rational
languages in a concise way, from which acceptors (automata) are built.
There are many largely used implementations, probably starting with Ken
Thompson \citep{thompson.68.cacm}, the creator of Unix, grep, etc.

There are numerous algorithms to build an automaton from an expression.  We
are particularly interested in the derivative-based family of algorithms
\citep{brzozowski.64.jacm, antimirov.1996.tcs, lombardy.2005.tcs,
  caron.2011.lata.2, demaille.16.arxiv}, because they offer a very natural interpretation to
states (they are labeled by an expression that denotes the future of the
states, i.e., the language/series accepted from this state).  This allowed
to support several extensions: extended operators (intersection, complement)
\citep{brzozowski.64.jacm, caron.2011.lata.2}, weights
\citep{lombardy.2005.tcs}, additional products (shuffle, infiltration), etc.

Multitape automata, including transducers, share many properties with
``single-tape'' automata, in particular the Fundamental Theorem
\citep[Theorem~2.1, p.~409]{sakarovitch.09.eat}: under appropriate
conditions, multitape automata and rational (multitape) series share the
same expressive power.  However, as far as the author knows, there is no
definition of multitape rational expressions that allows expressions such as
$\Ed_2 \coloneqq \EdTwo$ (\cref{ex:e2}).  To denote such a binary relation
between words, one had to build a (usual) rational expression in ``normal
form'', without tupling of expressions but only tuples of letters such as a
set of generators.  So for instance instead of $\Ed_2$, one must use
$\Ed_2' \coloneqq \paren{(a\tuple \eword)^+(\eword\tuple x) +
  (b\tuple\eword)^+(\eword\tuple y)}^*$,
which is larger, as is its derived-term automaton.

The contributions of this paper are twofold: we define (weighted) multitape
rational expressions featuring a $\tuple$ operator, and we provide an
algorithm to build an equivalent automaton.  This algorithm is a
generalization of the derived-term based algorithms, freed from the
requirement that the monoid is free.

\longskip

We first settle the notations in \cref{sec:notations}, provide an algorithm
to compute the expansion of an expression in \cref{sec:expa-of-expr}, which
is used in \cref{sec:expaton} to propose an alternative construction of the
derived-term automaton.

\longskip

The constructs exposed in this paper are implemented in
\vcsn\footnote{\label{foot:url}See the interactive environment,
  \url{http://vcsn-sandbox.lrde.epita.fr}, or its documentation,
  \url{http://vcsn.lrde.epita.fr/dload/2.3/notebooks/expression.derived_term.html},
  or this paper's companion notebook,
  \url{http://vcsn.lrde.epita.fr/dload/2.3/notebooks/CIAA-2016.html}.}.
\vcsn is a free-software platform dedicated to weighted automata and
rational expressions \citep{demaille.13.ciaa.2}; its lowest layer is a \Cxx
library, on top of which Python/IPython bindings provide an interactive
graphical environment.

\section{Notations}
\label{sec:notations}

\newcommand{\OB}[2]{\overbrace{#2}^{\text{#1}}}
\newcommand{\POB}[2]{\OB{\makebox[0pt]{#1}}{\vphantom{\bra{2}}#2}}
\newcommand{\UB}[2]{\underbrace{#2}_{\text{#1}}}
\newcommand{\PUB}[2]{\UB{\makebox[0pt]{#1}}{\vphantom{\bra{2}}#2}}
\newcommand{\textstack}[2]{
  \begin{tabular}{c}
    #1\\#2
  \end{tabular}
}

Our purpose is to define (weighted) multitape rational expressions, such as
$\Ed_1 \coloneqq \EdOne$ (weights are written in angle brackets).  It
relates $ade$ with $x$, with weight 4.  We introduce an algorithm to build a
multitape automaton (aka \dfn{transducer}) from such an expression, e.g.,
\cref{fig:aut:e1}.  This algorithm relies on \emph{rational expansions}.
They are to the derivatives of rational expressions what differential forms
are to the derivatives of functions.  Defining expansions requires several
concepts, defined bottom-up in this section.  The following figure presents
these different entities, how they relate to each other, and where we are
heading to: given a weighted multitape rational expression such as $\Ed_1$,
compute \emph{its} expansion:
\begin{align*}
\UB{Expansion (\cref{sec:expa})}
{
  \PUB{~~~~~~~~~~~Constant term}
  {
    \vphantom{\PUB{Dummy}{\POB{Weight}{\biggl[bra{5}\biggr]}}}
    \POB{Weight}{\bra{5}}
  }
  \;
  \oplus
  \PUB{Proper part of the expansion}
  {
    \UB{First}{
      \vphantom{\biggl[\biggr]}\POB{Label}{a|x}
    }
    \odot
    \biggl[
      \Lmul{2}{\PUB{Derived term}{\vphantom{\biggl[\biggr]}\POB{\textstack{Expression}{(\cref{sec:expr})}}{ce^*|y}}}
      \;\oplus\;\;
      \POB{Monomial}{
        \Lmul{4}{de^*|\und}
        }
    \biggr]
    \;\;\oplus\;\;
    b|x \odot
    \biggl[
      \OB{Polynomial (\cref{sec:poly})}
      {
        \Lmul{6}{\vphantom{\biggl[\biggr]}ce^*|y}
        \oplus
        \Lmul{3}{de^*|\und}
      }
    \biggr]
  }
}
\end{align*}
from which we build its derived-term automaton (\cref{fig:aut:e1}).

\begin{figure}[t]
  \centerline{\includegraphics[scale=.8]{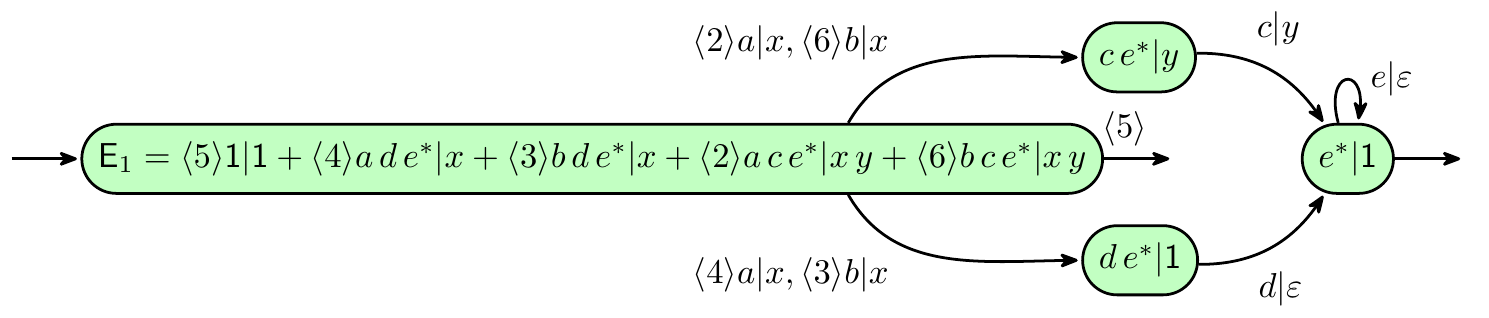}}
  \vspace{-3ex}
  \caption[The derived-term automaton of $\Ed_1$]{The derived-term
    automaton of $\Ed_1$ (see \cref{ex:e1,ex:e1:xpn,ex:e1:aut}) with\\
    $\Ed_1 \coloneqq \EdOne$.}
  \label{fig:aut:e1}
\end{figure}

It is helpful to think of expansions as a normal form for expressions.

\subsection{Rational Series}

Series will be used to define the semantics of the forthcoming structures:
they are to weighted automata what languages are to Boolean automata. Not
all languages are rational (denoted by an expression), and similarly, not
all series are rational (denoted by a weighted expression).  We follow
\citet[Chap.~III]{sakarovitch.09.eat}.

In order to cope with (possibly) several tapes, we cannot rely on the
traditional definitions based on the free monoid $A^*$ for some alphabet
$A$.

\subparagraph{Labels}%
Let $M$ be a monoid (e.g., $A^*$ or $A^* \times B^*$), whose neutral element
is denoted $\eword[M]$, or $\eword$ when clear from the context.  For
consistency with the way transducers are usually represented, we use
$m \tuple n$ rather than $(m, n)$ to denote the pair of $m$ and $n$.  For
instance $\eword[A^* \times B^*] = \eword[A^*] \tuple \eword[B^*]$, and
$\eword[M] \tuple a \in M \times \{a\}^*$.  A \dfn{set of generators} $G$ of
$M$ is a subset of $M$ such that $G^* = M$.  A monoid $M$ is of \dfn{finite
  type} (or \dfn{finitely generated}) if it admits a finite set of
generators.  A monoid $M$ is \dfn{graded} if it admits a \dfn{gradation}
function $\length{\cdot} \in M \rightarrow \N$ such that
$\forall m, n \in M$, $\length{m} = 0$ iff $m = \eword$, and
$\length{mn} = \length{m} + \length{n}$.  Cartesian products of graded
monoids are graded, and Cartesian products of finitely generated monoids are
finitely generated.  Free monoids and Cartesian products of free monoids are
graded and finitely generated.

\subparagraph{Weights}%
Let $\bra{\K, +, \cdot, \zeK, \unK}$ (or $\K$ for short) be a semiring whose
(possibly non commutative) multiplication will be denoted by juxtaposition.
$\K$ is \dfn{commutative} if its multiplication is.
$\K$ is a \dfn{topological semiring} if it is equipped with a topology, and
both addition and multiplication are continuous.  It is \dfn{strong} if the
product of two summable families is summable.

\subparagraph{Series}%
A (formal power) \dfn{series} over $M$ with \dfn{weights} (or
\dfn{multiplicities}) in $\K$ is a map from $M$ to $\K$.  The weight of
$m \in M$ in a series $s$ is denoted $s(m)$.  The \dfn{null} series,
$m \mapsto \zeK$, is denoted $0$; for any $m \in M$ (including $\eword[M]$),
$m$ denotes the series $u \mapsto \unK \text{ if $u = m$}, \zeK \text{
  otherwise}$.
If $M$ is of finite type, then we can define the Cauchy product of series.
$s \cdot t \coloneqq m \mapsto \sum_{u, v\in M \mid u v = m} s(u) \cdot
t(v)$.
Equipped with the pointwise addition
($s + t \coloneqq m \mapsto s(m) + t(m)$) and $\cdot$ as multiplication, the
set of these series forms a semiring denoted
$\bra{\SRk{M}, +, \cdot, 0, \eword}$.

The \dfn{constant term} of a series $s$, denoted $s_{\eword}$, is
$s(\eword)$, the weight of the empty word.  A series $s$ is \dfn{proper} if
$s_{\eword} = \zeK$. The \dfn{proper part} of $s$ is the proper series $s_p$
such that $s = s_{\eword} + s_p$.

\subparagraph{Star}%
The \dfn{star} of a series is an infinite sum:
$s^* \coloneqq \sum_{n\in\N}s^n$.  To ensure semantic soundness, we need $M$
to be graded monoid and $\K$ to be a strong topological semiring.

\begin{proposition}
  \label{prop:dev}
  Let $M$ be a graded monoid and $\K$ a strong topological semiring. Let
  $s \in \SRk{M}$, $s^*$ is defined iff $s_{\eword}^*$ is defined and then
  $s^* = s_{\eword}^* + s_{\eword}^*s_ps^*$.
\end{proposition}
\begin{proof}
  By \citep[Prop.~2.6, p.~396]{sakarovitch.09.eat} $s^*$ is defined iff
  $s_{\eword}^*$ is defined and then
  $s^* = (s_{\eword}^*s_p)^*s_{\eword}^* = s_{\eword}^*(s_ps_{\eword}^*)^*$.  The
  result then follows directly from $s^* = \eword + ss^*$:
  $s^* = s_{\eword}^*(s_ps_{\eword}^*)^* = s_{\eword}^*(\eword +
  (s_ps_{\eword}^*)(s_ps_{\eword}^*)^*) = s_{\eword}^* +
  s_{\eword}^*s_p(s_{\eword}^*(s_ps_{\eword}^*)^*) = s_{\eword}^* +
  s_{\eword}^*s_ps^*$.\qed
\end{proof}

\subparagraph{Tuple}%
We suppose $\K$ is commutative.  The \dfn{tupling} of two series
$s \in \SRk{M}, t \in \SRk{N}$, is the series
$s \tuple t \coloneqq m \tuple n \in M \times N \mapsto s(m)t(n)$.  It is a
member of $\SRk{M \times N}$.

\begin{proposition}
  \label{prop:series}
  For all series $s, s' \in \SRk{M}$ and $t, t' \in \SRk{N}$,
  $(s+s')\tuple t = s\tuple t + s'\tuple t$ and
  $s\tuple (t+t') = s\tuple t + s\tuple t'$.
\end{proposition}

\begin{proof}
  Let $m\tuple n \in M\times N$.
  $((s+s')\tuple t)(m\tuple n) = (s+s')(m) \cdot t(n) = (s(m)+s'(m)) \cdot
  t(n) = s(m)\cdot t(n)+s'(m)\cdot t(n) = (s\tuple t)(m\tuple n) \cdot
  (s'\tuple t)(m\tuple n) = (s\tuple t + s'\tuple t)(m\tuple n)$.
  Likewise for right distributivity.  \qed
\end{proof}

From now on, $M$ is a graded monoid of finite type, and $\K$ a commutative
strong topological semiring.

\subsection{Weighted Rational Expressions}
\label{sec:expr}

Contrary to the usual definition, we do not require a finite alphabet: any
set of generators $G \subseteq M$ will do.  For expressions with more than
one tape, we required $\K$ to be commutative; however, for single tape
expressions, our results apply to non-commutative semirings, hence there are
two exterior products.

\newcommand{\lmid}{\;\mid\;}
\begin{definition}[Expression]
  A \dfn{rational expression} $\Ed$ over $G$ is a term built from the
  following grammar, where $a \in G$ denotes any non empty label, and
  $k \in \K$ any weight:
  \begin{math}
    \Ed \Coloneqq \zed
          \lmid \und
          \lmid a
          \lmid \Ed + \Ed
          \lmid \lmul{k}{\Ed}
          \lmid \rmul{\Ed}{k}
          \lmid \Ed \cdot \Ed
          \lmid \Ed^*
          \lmid \Ed \tuple \Ed
  \end{math}.
\end{definition}

Expressions are syntactic; they are finite notations for (some) series.
\begin{definition}[Series Denoted by an Expression]
  Let $\Ed$ be an expression.  The series denoted by $\Ed$, noted
  $\sem{\Ed}$, is defined by induction on $\Ed$:
  \begin{gather*}
    \sem{\zed} \coloneqq 0 \qquad
    \sem{\und} \coloneqq \eword    \qquad
    \sem{a}    \coloneqq a \qquad
    \sem{\Ed+\Fd}       \coloneqq \sem{\Ed} + \sem{\Fd} \qquad
    \sem{\lmul{k}{\Ed}} \coloneqq {k}{\sem{\Ed}}   \\
    \sem{\rmul{\Ed}{k}} \coloneqq {\sem{\Ed}}{k}    \quad
    \sem{\Ed \cdot \Fd} \coloneqq \sem{\Ed} \cdot \sem{\Fd} \quad
    \sem{\Ed^*} \coloneqq \sem{\Ed}^* \quad
    \sem{\Ed \tuple \Fd} \coloneqq \sem{\Ed} \tuple \sem{\Fd}
  \end{gather*}
\end{definition}
An expression is \dfn{valid} if it denotes a series.  More specifically,
there are two requirements.  First, the expression must be well-formed,
i.e., concatenation and disjunction must be applied to expressions of
appropriate number of tapes.  For instance, $a+b|c$ and $a(b|c)$ are
ill-formed, $(a\tuple b)^* \tuple c + a\tuple (b\tuple c)^*$ is well-formed.
Second, to ensure that $\sem{\Fd}^*$ is well defined for each subexpression
of the form $\Fd^*$, the constant term of $\sem{\Fd}$ must be
\emph{starrable} in $\K$ (\cref{prop:dev}).  This definition, which involves
series (semantics) to define a property of expressions (syntax), will be
made effective (syntactic) with the appropriate definition of the constant
term $\ec(\Ed)$ \emph{of an expression} $\Ed$ (\cref{def:expa-of-expr}).

Let $[n]$ denote $\{1, \ldots, n\}$). The \dfn{size} (aka \dfn{length}) of a
(valid) expression $\Ed$, $\length{\Ed}$, is its total number of symbols,
not counting parenthesis; for a given tape number $i\in [k]$ the \dfn{width
  on tape $i$}, $\width{\Ed}_i$, is the number of occurrences of labels on
the tape $i$, the \dfn{width} of $\Ed$ (aka \dfn{literal length}),
$\width{\Ed} \coloneqq \sum_{i\in [k]}\width{\Ed}_i$ is the total number of
occurrences of labels.

Two expressions $\Ed$ and $\Fd$ are \dfn{equivalent} iff
$\sem{\Ed} = \sem{\Fd}$.  Some expressions are ``trivially equivalent''; any
candidate expression will be rewritten via the following \dfn{trivial
  identities}.  Any subexpression of a form listed to the left of a
`$\Rightarrow$' is rewritten as indicated on the right.
\begin{gather*}
  \Ed+\zed  \Rightarrow \Ed
  \ee
  \zed+\Ed  \Rightarrow \Ed
  \\
  \begin{aligned}[t]
    \lmul{\zeK}{\Ed} & \Rightarrow \zed &
    \lmul{\unK}{\Ed} & \Rightarrow \Ed  &
    \lmul{k}{\zed}   & \Rightarrow \zed &
    \lmul{k}{\lmul{h}{\Ed}} &\Rightarrow \lmul{kh}{\Ed}
    \\
    \rmul{\Ed}{\zeK} & \Rightarrow \zed &
    \rmul{\Ed}{\unK} & \Rightarrow  \Ed &
    \rmul{\zed}{k}   & \Rightarrow \zed &
    \rmul{\rmul{\Ed}{k}}{h}  &\Rightarrow \rmul{\Ed}{kh}
  \end{aligned}\\
  \rmul{(\lmul{k}{\Ed})}{h} \Rightarrow \lmul{k}{(\rmul{\Ed}{h})} \ee
  \rmul{\ell}{k} \Rightarrow \lmul{k}{\ell}
  \\ 
  \Ed \cdot \zed  \Rightarrow \zed \ee
  \zed \cdot \Ed  \Rightarrow \zed
  \\
  (\lmulq{k}{\und}) \cdot \Ed   \Rightarrow  \lmulq{k}{\Ed}
  \ee
  \Ed \cdot (\lmulq{k}{\und})   \Rightarrow  \rmulq{\Ed}{k}
  \\ 
  \zed^\star \Rightarrow \und
  \\
  (\lmulq{k}{\Ed}) \tuple (\lmulq{h}{\Fd}) \Rightarrow \lmulq{kh}{\Ed\tuple\Fd}
\end{gather*}
where $\Ed$ is a rational expression, $\ell \in G \cup \{\und\}$ a label,
$k, h\in \K$ weights, and $\lmulq{k}{\ell}$ denotes either $\lmul{k}{\ell}$,
or $\ell$ in which case $k = \unK$ in the right-hand side of $\Rightarrow$.
The choice of these identities is beyond the scope of this paper (see
\cite{sakarovitch.09.eat}), however note that they are limited to trivial
properties; in particular \dfn{linearity} (``weighted ACI'': associativity,
commutativity and
$\lmul{k}{\Ed} + \lmul{h}{\Ed} \Rightarrow \lmul{k+h}{\Ed}$) is not
enforced.  In practice, additional identities help reducing the automaton
size \citep{owens.2009.jfp}.

\subsection{Rational Polynomials}
\label{sec:poly}

At the core of the idea of ``partial derivatives'' introduced by
\citet{antimirov.1996.tcs}, is that of \emph{sets} of rational expressions,
later generalized in \emph{weighted sets} by \citet{lombardy.2005.tcs},
i.e., functions (partial, with finite domain) from the set of rational
expressions into $\K \setminus \{\zeK\}$.  It proves useful to view such
structures as ``polynomials of expressions''.  In essence, they capture the
linearity of addition.

\begin{definition}[Rational Polynomial]
  A \dfn{polynomial} (of rational expressions) is a finite (left) linear
  combination of expressions.  Syntactically it is a term built from the
  grammar
  \begin{math}
    \Pd \Coloneqq 0
    \mid \Lmul{k_1}{\Ed_1} \oplus \cdots \oplus \bra{k_n} \odot \Ed_n
  \end{math}
  where $k_i\in \K\setminus \{\zeK\}$ denote \emph{non-null} weights, and
  $\Ed_i$ denote \emph{non-null} expressions.  Expressions may not appear
  more than once in a polynomial.  A \dfn{monomial} is a pair
  $\bra{k_i} \odot \Ed_i$.
\end{definition}

We use specific symbols ($\odot$ and $\oplus$) to clearly separate the outer
polynomial layer from the inner expression layer.
Let $\Pd = \bigoplus_{i \in [n]}\Lmul{k_i}{\Ed_i}$ be a polynomial of
expressions.  The ``\dfn{projection}'' of $\Pd$ is the expression
$\expr{\Pd} \coloneqq \lmul{k_1}{\Ed_1} + \cdots + \lmul{k_n}{\Ed_n}$ (or
$\zed$ if $\Pd$ is null); this operation is performed on a canonical form of
the polynomial (expressions are sorted in a well defined order).
Polynomials denote series: $\sem{\Pd} \coloneqq \sem{\expr{\Pd}}$.  The
\dfn{terms} of $\Pd$ is the set
$\exprs{\Pd} \coloneqq \{\Ed_1, \ldots, \Ed_n\}$.

\begin{example}
  \label{ex:e1}
  Let $\Ed_1 \coloneqq \EdOne$.  Polynomial
  `$\Pd_{1,a\tuple x} \coloneqq \Lmul{2}{ce^*\tuple y} \oplus
  \Lmul{4}{de^*\tuple \und}$'
  has two monomials: `$\Lmul{2}{ce^*\tuple y}$' and
  `$\Lmul{4}{de^*\tuple \und}$'.  It denotes the (left) quotient of
  $\sem{\Ed_1}$ by $a\tuple x$, and
  `$\Pd_{1,b\tuple x} \coloneqq \Lmul{6}{ce^*\tuple y} \oplus
  \Lmul{3}{de^*\tuple \und}$' the quotient by $b\tuple x$.
\end{example}

Let
$\Pd = \bigoplus_{i \in [n]}\Lmul{k_i}{\Ed_i}, \Qd = \bigoplus_{j \in
  [m]}\Lmul{h_i}{\Fd_i}$
be polynomials, $k$ a weight and $\Fd$ an expression, all possibly null, we
introduce the following operations:
\begin{gather*}
  \Pd\cdot\Fd \coloneqq \bigoplus_{i \in [n]} \Lmul{k_i}{(\Ed_i\cdot\Fd)}
  \e
  \lmul{k}{\Pd} \coloneqq \bigoplus_{i \in [n]} \Lmul{kk_i}{\Ed_i}
  \e
  \rmul{\Pd}{k} \coloneqq \bigoplus_{i \in [n]} \Lmul{k_i}{(\rmul{\Ed_i}{k})}
  \\
  \Pd \tuple \und \coloneqq
  \bigoplus_{i \in [n]} \Lmul{k_i}{\Ed_i \tuple \und}
  \qquad
  \und \tuple \Pd \coloneqq
  \bigoplus_{i \in [n]} \Lmul{k_i}{\und \tuple \Ed_i}
  \\
  \Pd \tuple \Qd \coloneqq
  \bigoplus_{(i,j) \in [n] \times [m]} \Lmul{k_i\cdot h_j}{\Ed_i \tuple \Fd_j}
\end{gather*}
Trivial identities might simplify the result.
Note the asymmetry between left and right exterior products.  The addition
of polynomials is commutative, multiplication by zero (be it an expression
or a weight) evaluates to the null polynomial, and the left-multiplication
by a weight is distributive.

\begin{lemma}
  \label{lem:poly:ops}%
  \begin{math}
    \sem{\Pd\cdot\Fd} = \sem{\Pd} \cdot \sem{\Fd}
    \ee
    \sem{\lmul{k}{\Pd}} = \lmul{k}{\sem{\Pd}}
    \ee
    \sem{\rmul{\Pd}{k}} = \rmul{\sem{\Pd}}{k}
  \end{math}
  \\
  \begin{math}
    \sem{\Pd \tuple \Qd} = \sem{\Pd} \tuple \sem{\Qd}
  \end{math}.
\end{lemma}
\begin{longenv}
  \begin{proof}
    See \cref{proof:lem:poly:ops}.
  \end{proof}
\end{longenv}

\subsection{Rational Expansions}
\label{sec:expa}

\begin{definition}[Rational Expansion]
  A \dfn{rational expansion} $\Xd$ is a term
  \begin{math}
    \Xd \Coloneqq \bra{\Xd_{\eword}} \oplus a_1 \odot[\Xd_{a_1}] \oplus \cdots \oplus a_n \odot[\Xd_{a_n}]
  \end{math}
  where $\Xd_{\eword} \in \K$ is a weight (possibly null),
  $a_i \in G\setminus \{\eword\}$ non-empty labels (occurring at most once),
  and $\Xd_{a_i}$ non-null polynomials.  The \dfn{constant term} is
  $\Xd_{\eword}$, the \dfn{proper part} is
  $\Xd_p \coloneqq a_1 \odot[\Xd_{a_1}] \oplus \cdots \oplus a_n
  \odot[\Xd_{a_n}]$,
  the \dfn{firsts} is $f(\Xd) \coloneqq \{a_1, \ldots, a_n\}$ (possibly
  empty) and the \dfn{terms}
  $\exprs{\Xd} \coloneqq \bigcup_{i\in [n]}\exprs{\Xd_{a_i}}$.
\end{definition}
To ease reading, polynomials are written in square brackets.  Contrary to
expressions and polynomials, there is no specific term for the null
expansion: it is represented by $\bra{\zeK}$, the null weight.  Except for
this case, null constant terms are left implicit.  Expansions will be
written:
\begin{math}
  \Xd = \bra{\Xd_{\eword}} \oplus \bigoplus_{a \in f(\Xd)} a \odot[\Xd_a]
\end{math}.  When more convenient, we write $\Xd(\ell)$ instead of
$\Xd_{\ell}$ for $\ell \in f(\Xd) \cup \{\eword\}$.

An expansion $\Xd$ can be ``projected'' as a rational expression
$\expr{\Xd}$ by mapping weights, labels and polynomials to their
corresponding rational expressions, and $\oplus$/$\odot$ to the
sum/concatenation of expressions.  Again, this is performed on a canonical
form of the expansion: labels are sorted.  Expansions also denote series:
$\sem{\Xd} \coloneqq \sem{\expr{\Xd}}$.  An expansion $\Xd$ is
\dfn{equivalent} to an expression $\Ed$ iff $\sem{\Xd} = \sem{\Ed}$.

\begin{example}[\cref{ex:e1} continued]
  \label{ex:e1:xpn}
  Expansion
  $\Xd_1 \coloneqq \bra{5} \oplus a|x \odot [\Pd_{1,a|x}] \oplus b|x \odot
  [\Pd_{1,b|x}]$
  has $\Xd_1(\eword) = \bra{5}$ as constant term, and maps the generator
  $a|x$ (resp.\ $b|x$) to the polynomial $\Xd_1(a|x) = \Pd_{1,a|x}$ (resp.\
  $\Xd_1(b|x) = \Pd_{1,b|x}$).  $\Xd_1$ can be proved to be equivalent to
  $\Ed_1$.
\end{example}

Let $\Xd, \Yd$ be expansions, $k$ a weight, and $\Ed$ an expression (all
possibly null):
\begin{gather}
  \label{eq:epn:plus:epn}
  \Xd \oplus \Yd
  \coloneqq
  \bra{\Xd_{\eword}+ \Yd_{\eword}}\oplus\bigoplus_{\mathclap{a \in f(\Xd) \cup f(\Yd)}} a \odot [\Xd_a \oplus \Yd_a]
  \\
  \lmul{k}{\Xd} \coloneqq \bra{k\Xd_{\eword}} \oplus \bigoplus_{\mathclap{a \in f(\Xd)}} a \odot [\lmul{k}{\Xd_a}]
  \ee
  \rmul{\Xd}{k} \coloneqq \bra{\Xd_{\eword} k} \oplus \bigoplus_{\mathclap{a \in f(\Xd)}} a \odot [\rmul{\Xd_a}{k}]
  \\
  \label{eq:epn:mul:epn}
  \Xd\cdot\Ed \coloneqq \bigoplus_{\mathclap{a \in f(\Xd)}} a \odot [\Xd_a \cdot \Ed] \ee \text{with $\Xd$ proper: $\Xd_{\eword} = \zeK$}
  \\
  \label{eq:epn:tuple:epn}
  \begin{split}
  \Xd\tuple\Yd &\coloneqq
  \bra{\Xd_{\eword} \Yd_{\eword}}
  \oplus \bra{\Xd_{\eword}} \bigoplus_{\mathclap{b\in f(\Yd)}} (\eword| b) \odot (\und \tuple \Yd_b)
  \oplus \bra{\Yd_{\eword}} \bigoplus_{\mathclap{a\in f(\Xd)}} (a| \eword) \odot (\Xd_a \tuple \und)
\\
  & \quad \oplus \bigoplus_{\mathclap{a|b \in f(\Xd) \times f(\Yd)}}(a| b) \odot (\Xd_a \tuple \Yd_b)
  \end{split}
\end{gather}
Since by definition expansions never map to null polynomials, some firsts
might be smaller that suggested by these equations.  For instance in
$\mathbb{Z}$ the sum of $\bra{1} \oplus a \odot [\Lmul{1}{b}]$ and
$\bra{1} \oplus a \odot [\Lmul{-1}{b}]$ is $\bra{2}$.

The following lemma is simple to establish: lift semantic equivalences, such
as \cref{prop:series}, to syntax, using \cref{lem:poly:ops}.
\begin{lemma}
  \label{lem:xpn:semantics}%
  \begin{math}
    \sem{\Xd \oplus \Yd} = \sem{\Xd} + \sem{\Yd} \ee
    \sem{\lmul{k}{\Xd}} = \lmul{k}{\sem{\Xd}} \ee
    \sem{\rmul{\Xd}{k}} = \rmul{\sem{\Xd}}{k}
  \end{math}\\
  \begin{math}
    \sem{\Xd\cdot\Ed} = \sem{\Xd}\cdot\sem{\Ed}
    \ee
    \sem{\Xd \tuple \Yd} = \sem{\Xd} \tuple \sem{\Yd}
  \end{math}
\end{lemma}

\subsection{Finite Weighted Automata}

\begin{definition}[Weighted Automaton]
  A \dfn{weighted automaton} $\Ac$ is a tuple $\bra{M, G, \K, Q, E, I, T}$
  where:
  \begin{itemize}
  \item $M$ is a monoid,
  \item $G$ (the labels) is a set of generators of $M$,
  \item $\K$ (the set of weights) is a semiring,
  \item $Q$ is a finite set of states,
  \item $I$ and $T$ are the \dfn{initial} and \dfn{final} functions
    from $Q$ into $\K$,
  \item $E$ is a (partial) function from $Q \times G \times Q$ into
    $\K \setminus \{\zeK\}$;

    its domain represents the transitions:
    $(\mathit{source}, \mathit{label}, \mathit{destination})$.
  \end{itemize}
  An automaton is \dfn{proper} if no label is $\eword_M$.
\end{definition}

A \dfn{computation}
$p = (q_0, a_0, q_1)(q_1, a_1, q_2)\cdots(q_n, a_n, q_{n+1})$ in an
automaton is a sequence of transitions where the source of each is the
destination of the previous one; its \dfn{label} is
$a_0a_1\cdots a_n \in M$, its \dfn{weight} is
$I(q_0) \otimes E(q_0, a_0, q_1) \otimes \cdots \otimes E(q_n, a_n, q_{n+1})
\otimes T(q_{n+1}) \in \K$.
The \dfn{evaluation} of word $u$ by $\Ac$, $\Ac(u)$, is the sum of the
weights of all the computations labeled by $u$, or $\zeK$ if there are none.
The \dfn{behavior} of an automaton $\Ac$ is the series
$\sem{\Ac} \coloneqq m \mapsto \Ac(m)$.  A state $q$ is \dfn{initial} if
$I(q) \neq \zeK$.  A state $q$ is \dfn{accessible} if there is a computation
from an initial state to $q$.  The \dfn{accessible} part of an automaton
$\Ac$ is the subautomaton whose states are the accessible states of $\Ac$.
The size of an automaton, $\length{\Ac}$, is its number of states.

\smallskip

We are interested, given an expression $\Ed$, in an algorithm to compute an
automaton $\Ac_\Ed$ such that $\sem{\Ac_\Ed} = \sem{\Ed}$
(\cref{def:expaton}).  To this end, we first introduce a simple recursive
procedure to compute \emph{the} expansion of an expression.

\section{Expansion of a Rational Expression}
\label{sec:expa-of-expr}
\begin{definition}[Expansion of a Rational Expression]
  \label{def:expa-of-expr}
  The \dfn{expansion of a rational expression} $\Ed$, written $d(\Ed)$, is
  defined inductively as follows:
  \begin{gather}
    \label{eq:epn:cst}
    d(\zed) \coloneqq \bra{\zeK} \ee
    d(\und) \coloneqq \bra{\unK} \ee
    d(a)    \coloneqq a \odot [\Lmul{\unK}{\und}]
    \\
    \label{eq:epn:add}
    d(\Ed+\Fd) \coloneqq d(\Ed) \oplus d(\Fd)
    \\
    \label{eq:epn:kmul}
    d(\lmul{k}{\Ed}) \coloneqq \lmul{k}{d(\Ed)} \ee
    d(\rmul{\Ed}{k}) \coloneqq \rmul{d(\Ed)}{k}
    \\
    \label{eq:epn:mul}
    d(\Ed \cdot \Fd)  \coloneqq \ep(\Ed)\cdot \Fd \oplus \lmul{\ec(\Ed)}{d(\Fd)}
    \\
    \label{eq:epn:star}
    d(\Ed^*) \coloneqq \bra{\ec(\Ed)^*} \oplus \lmul{\ec(\Ed)^*}{\ep(\Ed) \cdot \Ed^*}
    \\
    \label{eq:epn:tuple}
    d(\Ed\tuple \Fd) \coloneqq d(\Ed) \tuple d(\Fd)
  \end{gather}
  where $\ec(\Ed) \coloneqq d(\Ed)_{\eword}, \ep(\Ed) \coloneqq d(\Ed)_p$ are
  the constant term/proper part of $d(\Ed)$.
\end{definition}

The right-hand sides are indeed expansions.  The computation
trivially terminates: induction is performed on strictly smaller
subexpressions.  These formulas are enough to compute the expansion of an
expression; there is no secondary process to compute the firsts --- indeed
$d(a) \coloneqq a \odot [\Lmul{\unK}{\und}]$ suffices and every other case
simply propagates or assembles the firsts --- or the constant terms.
\begin{longenv}

  Of course, in an implementation, a single recursive call to $d(\Ed)$ is
  performed for \cref{eq:epn:mul,eq:epn:star}, from which $\ec(\Ed)$ and
  $\ep(\Ed)$ are obtained, and additional expansions are computed only when
  needed.  So they should rather be written:
  \begin{align*}
    d(\Ed\cdot \Fd)
    & \coloneqq
      \mathtt{let}\;
      \Xd = d(\Ed)\;\mathtt{in}\;
      \mathtt{if}\; \bra{\Xd_{\eword}} \ne \zeK
      \;\mathtt{then}\;
      \Xd_p \cdot \Fd \oplus \lmul{\Xd_{\eword}}{d(\Fd)}
      \;\mathtt{else}\;
      \Xd_p \cdot \Fd
    \\
    d(\Ed^*)
    & \coloneqq
      \mathtt{let}\;
      \Xd = d(\Ed)\;\mathtt{in}\;
      \bra{\Xd_{\eword}^*}
      \oplus \lmul{\Xd_{\eword}^*}{\Xd_p \cdot \Ed^*}
  \end{align*}
  Besides, existing expressions should be referenced to, not
  duplicated.  In the previous piece of code, $\Ed^*$ is not built again,
  the input argument is reused.

\end{longenv}
Note that the firsts are a subset of the labels of the expression, hence of
$G\setminus \{\eword\}$.  In particular, no first includes $\eword$.

\begin{proposition}
  \label{prop:xpn:equiv}
  The expansion of a rational expression is equivalent to the expression.
\end{proposition}
\begin{proof}
  We prove that $\sem{d(\Ed)} = \sem{\Ed}$ by induction on the expression.
  The equivalence is straightforward for
  \cref{eq:epn:cst,eq:epn:add,eq:epn:kmul,eq:epn:tuple}, viz.,
  \newcommand{\since}[1]{$ (#1) $}
  \begin{math}
    \sem{d(\Ed \tuple \Fd)}
    = \sem{d(\Ed) \tuple d(\Fd)}        \since{by \cref{eq:epn:tuple}}
    = \sem{d(\Ed)} \tuple \sem{d(\Fd)}  \since{by \cref{lem:xpn:semantics}}
    = \sem{\Ed} \tuple \sem{\Fd}        \since{by induction hypothesis}
    = \sem{\Ed\tuple \Fd}               \since{by \cref{lem:xpn:semantics}}
  \end{math}.
  The case of multiplication, \cref{eq:epn:mul}, follows from:
  \begin{align*}
    \sem{d(\Ed \cdot \Fd)}
    &= \sem{\ep(\Ed) \cdot \Fd \oplus \bra{\ec(\Ed)} \cdot d(\Fd)} &
    &= \sem{\ep(\Ed)}\cdot\sem{\Fd} + \bra{\ec(\Ed)} \cdot \sem{d(\Fd)} \\
    &= \sem{\ep(\Ed)}\cdot\sem{\Fd} + \bra{\ec(\Ed)} \cdot \sem{\Fd} &
    &= \paren{\sem{\bra*{\ec(\Ed)}}+\sem{\ep(\Ed)}} \cdot\sem{\Fd} \\
    &= \sem{\bra{\ec(\Ed)} + \ep(\Ed)} \cdot\sem{\Fd} &
    &= \sem{d(\Ed)} \cdot \sem{\Fd} \\
    &= \sem{\Ed} \cdot \sem{\Fd} &
    &= \sem{\Ed\cdot \Fd}
  \end{align*}
  It might seem more natural to exchange the two terms (i.e.,
  $\bra{\ec(\Ed)} \cdot d(\Fd) \oplus \ep(\Ed)\cdot \Fd$), but an
  implementation first computes $d(\Ed)$ and then computes $d(\Fd)$
  \emph{only if} $\ec(\Ed) \ne \zeK$.
  The case of Kleene star, \cref{eq:epn:star}, follows from \cref{prop:dev}.
  \qed
\end{proof}

\section{Expansion-Based Derived-Term Automaton}
\label{sec:expaton}


\vspace{-1ex}
\begin{definition}[Expansion-Based Derived-Term Automaton]
  \label{def:expaton}
  The \dfn{derived-term automaton} of an expression $\Ed$ over $G$ is the
  \emph{accessible part} of the automaton
  ${\Ac_\Ed} \coloneqq \bra{M, G, \K, Q, E, I, T}$ defined as follows:
  \begin{itemize}
  \item $Q$ is the set of rational expressions on alphabet $A$ with weights
    in $\K$,
  \item $I = \Ed \mapsto \unK$,
  \item
    $E(\Fd, a, \Fd') = k \text{ iff } a \in f(d(\Fd)) \;\mathrm{and}\;
    \lmul{k}{\Fd'} \in d(\Fd)(a)$,
  \item $T(\Fd) = k$ iff $\bra{k} = d(\Fd)(\eword)$.
  \end{itemize}
\end{definition}

Since the firsts exclude $\eword$, this automaton is proper.  It is
straightforward to extract an algorithm from \cref{def:expaton}, using a
work-list of states whose outgoing transitions to
compute\ifthen{\boolean{long}}{ (see \cref{sec:algo})}.  The
\cref{fig:expaton} illustrates the process.  This approach admits a natural
lazy implementation: the whole automaton is not computed at once, but
rather, states and transitions are computed on-the-fly, on demand, for
instance when evaluating a word \citep{demaille.16.arxiv}.  However, we must justify
\cref{def:expaton} by proving that this automaton is finite
(\cref{thm:size}).

\begin{figure}[t]
  \centering
  \includegraphics[scale=.8]{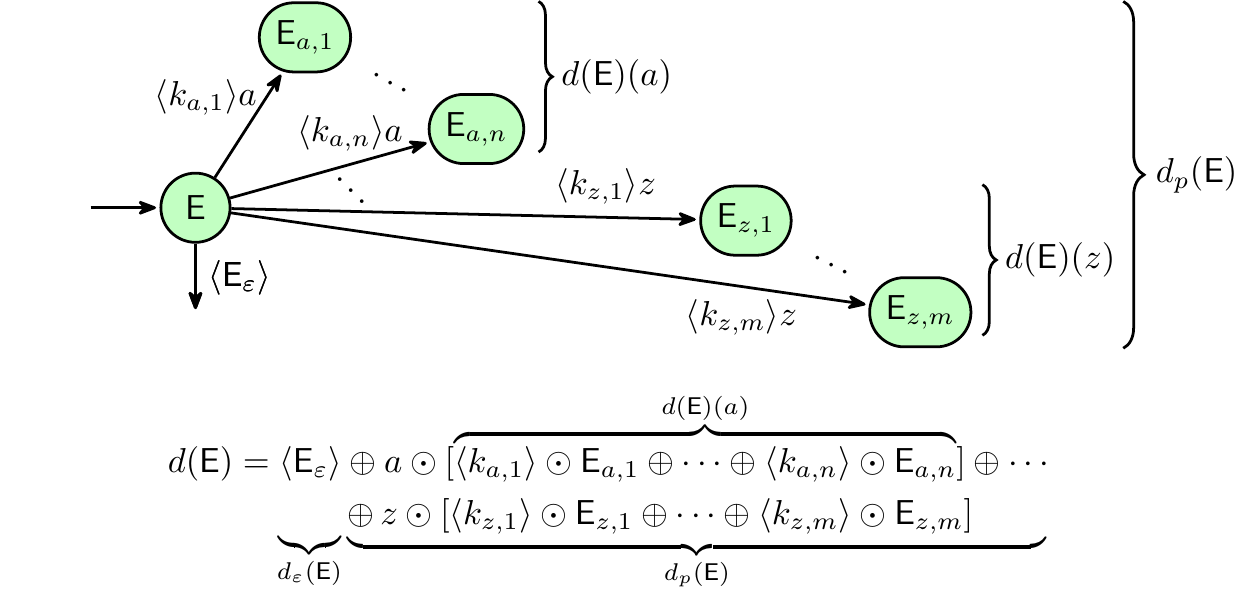}
  \vspace{-2ex}
  \caption[$\Ac_\Ed$]{Initial part of $\Ac_\Ed$, the derived-term automaton
    of $\Ed$.  This figure is somewhat misleading in that some
    $\di{\Ed}{a}{i}$ might be equal to an $\di{\Ed}{z}{j}$, or $\Ed$ (but
    never another $\di{\Ed}{a}{j}$).}
  \label{fig:expaton}
\end{figure}

\begin{example}[\cref{ex:e1,ex:e1:xpn} continued]
  \label{ex:e1:aut}
  With $\Ed_1 \coloneqq \EdOne$, one has:
  \begin{align*}
    d(\Ed_1)
    &= \bra{5}
    \oplus a|x \odot [\Lmul{2}{ce^* \tuple y} \oplus \Lmul{4}{de^* \tuple \eword}]
    \oplus b|x \odot [\Lmul{6}{ce^* \tuple y} \oplus \Lmul{3}{de^* \tuple \eword}]
    \\
    &= \Xd_1 \qquad \text{(from \cref{ex:e1:xpn})}
  \end{align*}
  \cref{fig:aut:e1} shows the resulting derived-term automaton.
\end{example}

\begin{theorem}
  \label{thm:size}
  For any $k$-tape expression $\Ed$,
  $\length{\Ac_\Ed} \le \prod_{i \in [k]}(\width{\Ed}_i+1) + 1$.
\end{theorem}

\begin{proof}
  \ifthen{\boolean{long}}{The detailed proof is available in
    \cref{sec:dt}. } The proof goes in several steps.  First introduce the
  \dfn{true derived terms} of $\Ed$, a set of expressions noted $\TD(\Ed)$,
  and the \dfn{derived terms} of $\Ed$,
  $\D(\Ed) \coloneqq \TD(\Ed) \cup \{\Ed\}$.  $\TD(\Ed)$ admits a simple
  inductive definition similar to \citep[Def.~3]{angrand.2010.jalc}, to
  which we add
  $\TD(\Ed\tuple \Fd) \coloneqq (\TD(\Ed) \tuple \TD(\Fd)) \cup (\{\und\}
  \tuple \TD(\Fd)) \cup (\TD(\Ed) \tuple \{\und\})$,
  where for two sets of expressions $E, F$ we introduce
  $E \tuple F \coloneqq \{\Ed \tuple \Fd\}_{(\Ed,\Fd)\in E\times F}$.
  Second, verify that
  $\length{\TD(\Ed)} \le \prod_{i \in [k]}(\width{\Ed}_i + 1)$ (hence
  finite).  Third, prove that $\D(\Ed)$ is ``stable by expansion'', i.e.,
  $\forall \Fd \in \D(\Ed), \exprs{d(\Fd)} \subseteq \D(\Ed)$.
  Finally, observe that the states of $\Ac_\Ed$ are therefore members of
  $\D(\Ed)$, whose size is less than or equal to $1 + \length{\TD(\Ed)}$.
  \qed
\end{proof}

\begin{theorem}
  \label{thm:expaton}%
  Any expression $\Ed$ and its expansion-based derived-term automaton
  $\Ac_\Ed$ denote the same series, i.e., $\sem{\Ac_\Ed} = \sem{\Ed}$.
\end{theorem}

\noindent
\begin{minipage}[c]{.45\linewidth}
  \begin{example}
    Let $\Ac_k$ be the derived-term automaton of the $k$-tape expression
    $a_1^* \tuple \cdots \tuple a_k^*$.  The states of $\Ac_k$ are all the
    possible expressions where the tape $i$ features $\und$ or $a_i^*$,
    except $\und \tuple \cdots \tuple \und$.  Therefore
    $\length{\Ac_k} = 2^k -1$, and
    $\prod_{i \in [k]}(\width{\Ed}_i+1) = 2^k$.

    \medskip

    $\Ac_3$, the derived-term automaton of $a^* \tuple b^* \tuple c^*$, is
    depicted on the right.
  \end{example}
\end{minipage}
\hfill
\begin{minipage}[c]{.45\linewidth}
  \raisebox{-\height}{\includegraphics[width=1\linewidth]{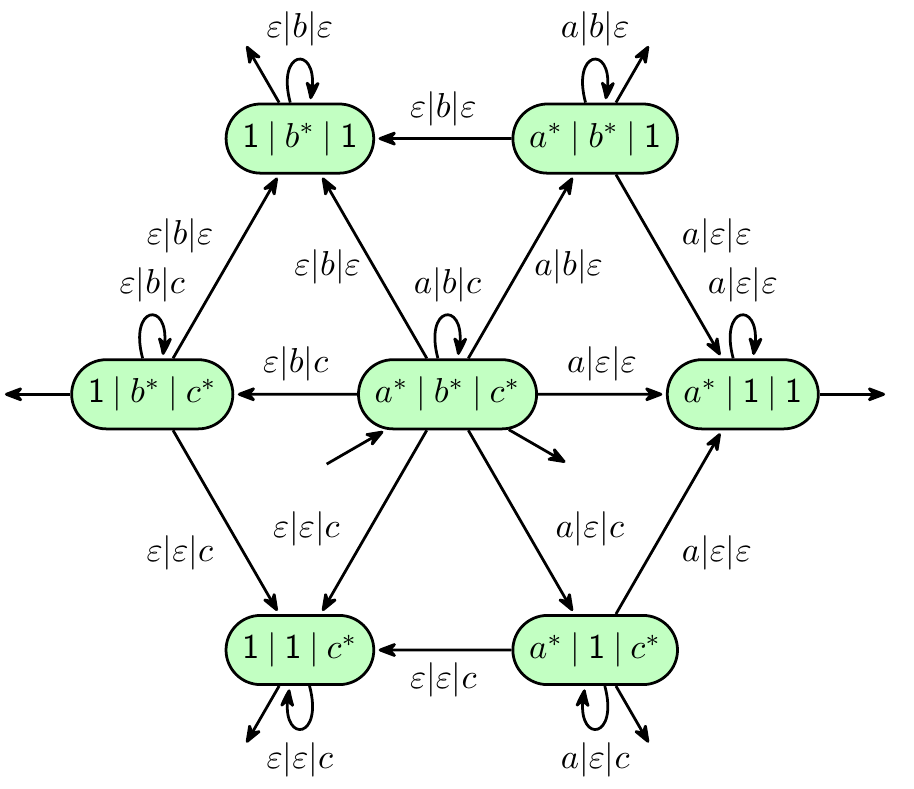}}
\end{minipage}

\vspace{-4mm}

\begin{proof}[\cref{thm:expaton}]
  We will prove $\sem{\Ac_\Ed}(m) = \sem{\Ed}(m)$ by induction on $m\in M$.
  If $m = \eword$, then
  $\sem{\Ac_\Ed}(m) = \Ed_{\eword} = d(\Ed)(\eword) = \sem{d(\Ed)}(\eword) =
  \sem{\Ed}(\eword)$.

  If $m$ is not $\eword$, then it can be generated in a (finite) number of
  ways: let
  $F(\Ed, m) \coloneqq \{(a, m_a) \in f(d(\Ed)) \times M \mid m = am_a\}$.
  $F(\Ed, m)$ is a function: for a given $a$, there is at most one $m_a$
  such that $(a, m_a) \in F(\Ed, m)$.  \cref{fig:expaton} is helpful.
  \abovedisplayskip=\abovedisplayshortskip
  \begin{align*}
    \sem{\Ac_\Ed}(m)
    &= \sum_{(a, m_a)\in F(\Ed, m)} \sum_{i \in[n_a]}\bra*{\di{k}{a}{i}}\sem*{\Ac_{\di{\Ed}{a}{i}}}(m_a)
    & \text{by definition of $\Ac_\Ed$} \\
    &= \sum_{(a, m_a)\in F(\Ed, m)} \sum_{i \in[n_a]}\bra*{\di{k}{a}{i}}\sem*{\di{\Ed}{a}{i}}(m_a)
    & \text{by induction hypothesis}\\
    &= \sum_{(a, m_a)\in F(\Ed, m)} \semBig{\sum_{i \in[n_a]}\bra*{\di{k}{a}{i}}\di{\Ed}{a}{i}}(m_a)
    & \text{by \cref{lem:poly:ops}}\\
    &= \sum_{(a, m_a)\in F(\Ed, m)} \sem{d(\Ed)(a)}(m_a) = \mathrlap{\sum_{(a, m_a)\in F(\Ed, m)} \sem{a\odot d(\Ed)(a)}(am_a)}\\
    &= \sum_{a\in f(d(\Ed))} \sem{a\odot d(\Ed)(a)}(m)
    & \text{$F(\Ed, m)$ is a function}\\
    &= \semBig{\sum_{a\in f(d(\Ed))} a\odot d(\Ed)(a)}(m)
    & \text{by \cref{lem:xpn:semantics}}\\
    &= \sem{\ec(\Ed)}(m)
    & \text{by definition} \\
    &= \sem{d(\Ed)}(m) & \text{since $m \ne \eword$} \\
    &= \sem{\Ed}(m) & \text{by \cref{prop:xpn:equiv}} && \text{\qed}
  \end{align*}
\end{proof}

\begin{example}
  \label{ex:e2}
  Let $\Ed_2 \coloneqq \EdTwo$, where $\Ed^+ \coloneqq\Ed\Ed^*$.  Its
  expansion is
  \abovedisplayskip=\abovedisplayshortskip
  \begin{align*}
    d(\Ed_2)
    & = \bra{1}
      \oplus a|x \odot
      \left[(a^* \tuple \und)
      ( a^+ \tuple x + b^+ \tuple y )^*\right]
      \oplus b|y \odot
      \left[(b^* \tuple \und)
      ( a^+ \tuple x + b^+ \tuple y )^*\right]\\
    &= \bra{1}
      \oplus a|x \odot \left[(a^* \tuple \und) \Ed_2\right]
      \oplus b|y \odot \left[(b^* \tuple \und) \Ed_2\right]
  \end{align*}
  Its derived-term automaton is:

  \vspace{-13mm}\hspace{2cm}
  \centerline{\includegraphics[scale=.8]{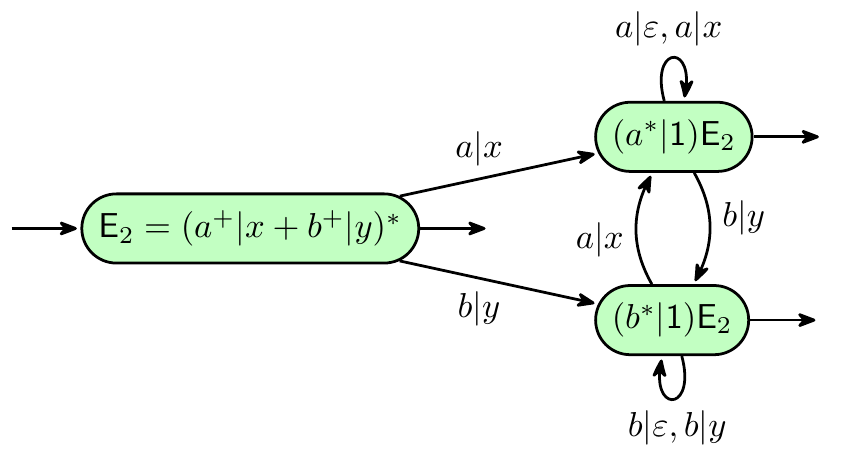}}
\end{example}


\vspace{-8ex}
\section{Related Work}
\label{sec:related}

Multitape rational expressions have been considered early
\citep{makarevskii.1969.csa}, but ``an n-way regular expression is simply a
regular expression whose terms are n-tuples of alphabetic symbols or
$\eword$'' \citep{kaplan.1994.cl}.  However, \citet{kaplan.1994.cl} do
consider the full generality of the semantics of operations on rational
languages and rational relations, including $\times$, the Cartesian product
of languages, and even use rational expressions more general than their
definition.  They do not, however, provide an explicit automaton
construction algorithm, apparently relying on the simple inductive
construction (using the Cartesian product between automata).  Our $\tuple$
operator on series was defined as the \dfn{tensor product}, denoted
$\otimes$, by \citet[Sec.~III.3.2.5]{sakarovitch.09.eat}, but without
equivalent for expressions.

\longskip

\citet{brzozowski.64.jacm} introduced the idea of derivatives of expressions
as a means to construct an equivalent automaton.  The method applies to
extended (unweighted) rational expressions, and constructs a deterministic
automaton.  \citet{antimirov.1996.tcs} modified the computation to rely on
parts of the derivatives (``partial derivatives''), which results in
nondeterministic automata.
\longpar
\citet{lombardy.2005.tcs} extended this approach to support weighted
expressions; independently, and with completely different foundations,
\citet{rutten.2003.tcs} proposed a similar construction.
\citet{caron.2011.lata.2} introduced support for (unweighted) extended
expressions.
\citet{demaille.16.arxiv} provides support for weighted extended
expressions; expansions, originally mentioned by
\citet{brzozowski.64.jacm}, are placed at the center of the construct,
replacing derivatives, to gain independence with respect to the size of
the alphabet, and efficiency.  However, the proofs still relied on
derivatives, contrary to the present work.

\longskip

\ifthenelse{\boolean{long}}%
{%
  \citet{makarevskii.1969.csa} define derivatives, but (i) in the case of
  expressions over tuples of letters, and (ii) only when in so-called
  ``standard form'', for which he notes ``no method of constructing [an]
  n-expression in standard form for a regular n-expression is known.''
  However, from \cref{eq:epn:tuple} one can deduce a definition of
  derivatives for the tuple operator (see \cref{sec:derivatives} for more
  details):
  \begin{align*}
    c(\Ed \tuple \Fd)     & \coloneqq c(\Ed) \cdot c(\Fd),
    & \derivative{a|b}{(\Ed \tuple \Fd)},
    & \coloneqq \da{\Ed}\tuple\db{\Fd},\\
    && \derivative{a|\eword}{(\Ed \tuple \Fd)},
    & \coloneqq \lmul{c(\Fd)}{(\da{\Ed} \tuple \und)}, \\
    &&\derivative{\eword|b}{(\Ed \tuple \Fd)},
    & \coloneqq \lmul{c(\Ed)}{(\und \tuple \db{\Fd})}.
  \end{align*}%
  From an implementation point of view, that would lead to repeated
  computations of $\derivative{a}{\Ed}$ and of $\derivative{b}{\Fd}$, unless
  one would cache them, but that's what expansions do.

  Note that these derivatives are no longer equivalent to the left quotient
  of the corresponding language.  Consider
  $\Fd \coloneqq (a^*\tuple \und)\EdTwo$: the language it denotes includes
  $ab|y$, yet $\derivative{a|y}{\Fd} = \bra{\zeK}$.  Albeit surprising, this
  result is nevertheless sufficient as can be observed in the derived-term
  automaton in \cref{ex:e2}: while the state $(a^*\tuple \und)\EdTwo$ does
  accept words starting with $a$ on the first tape, and $y$ on the second,
  an outgoing transition on $a|y$ would result in a more complex automaton.
}%
{%
  Based on \cref{eq:epn:tuple} one could attempt to define a
  derivative-based version, with
  \begin{math}
    \derivative{a|b}{(\Ed \tuple \Fd)} \coloneqq
    \derivative{a}{\Ed} \tuple \derivative{b}{\Fd}
  \end{math}, however this is troublesome on several regards.  First, it would
  also require $\derivative{a|\eword}{}$ and $\derivative{\eword|b}{}$, whose
  semantics is dubious.  Second, from an implementation point of view, that
  would lead to repeated computations of $\derivative{a}{\Ed}$ and of
  $\derivative{b}{\Fd}$, unless one would cache them, but that's exactly what
  expansions do.  And finally observe that in the derived-term automaton in
  \cref{ex:e2}, the state $(a^*\tuple \und)\EdTwo$ accepts words starting with
  $a$ on the first tape, and $y$ on the second, yet an outgoing transition on
  $a|y$ would result in a more complex automaton.
  \longpar
  Alternative definitions of derivatives may
  exist\footnote{\citet{makarevskii.1969.csa} define derivatives, but (i) in
    the case of expressions over tuples of letters, and (ii) only when in
    so-called ``standard form'', for which he notes ``no method of
    constructing [an] n-expression in standard form for a regular n-expression
    is known.''}, but anyway they would no longer be equivalent to taking the
  left-quotient of the corresponding language: $a|y$ \emph{is} a viable prefix
  from this state.%
}

\longskip

Different constructions of the derived-term automaton have been discovered
\citep{allauzen.2006.mfcs.2, champarnaud.2007.dlt.2}.  They do not rely on
derivatives at all.  It is an open question whether these approaches can be
adapted to support a tuple operator.

\section{Conclusion}

Our work is in the continuation of derivative-based computations of the
derived-term automaton \citep{brzozowski.64.jacm, antimirov.1996.tcs,
  lombardy.2005.tcs, caron.2011.lata.2}.  However, we replaced the
derivatives by expansions, which lifted the requirement for the monoid of
labels to be free.

In order to support $k$-tape (weighted) rational expressions, we introduced
a tupling operator, which is more compact and readable than simple
expressions on $k$-tape letters.  We demonstrated how to build the
derived-term automaton for any such expressions.

\vcsn{}\cref{foot:url} implements the techniques exposed in this paper.  Our
future work aims at other operators, and studying more closely the
complexity of the algorithm.
\begin{longenv}
  The usual state-elimination method to compute an expression from an
  automaton works perfectly, however we are looking for means to reduce the
  expression size.
\end{longenv}

\paragraph*{Acknowledgments} The author thanks the anonymous reviewers for
their constructive comments, and A.~Duret-Lutz, S.~Lombardy, L.~Saiu and
J.~Sakarovitch for their feedback during this work.

\bibliographystyle{myabbrvnat}
\bibliography{%
  article,%
  share/bib/acronyms,%
  share/bib/lrde,%
  share/bib/comp.compilers.automata%
}

\begin{thebibliography}{15}
\providecommand{\natexlab}[1]{#1}
\providecommand{\url}[1]{\texttt{#1}}
\expandafter\ifx\csname urlstyle\endcsname\relax
  \providecommand{\doi}[1]{doi: #1}\else
  \providecommand{\doi}{doi: \begingroup \urlstyle{rm}\Url}\fi

\bibitem[Allauzen and Mohri(2006)]{allauzen.2006.mfcs.2}
C.~Allauzen and M.~Mohri.
\newblock A unified construction of the {G}lushkov, follow, and {A}ntimirov
  automata.
\newblock In \emph{MFCS}, vol. 4162 of \emph{LNCS}, pp. 110--121. Springer,
  2006.

\bibitem[Angrand et~al.(2010)Angrand, Lombardy, and
  Sakarovitch]{angrand.2010.jalc}
P.-Y. Angrand, S.~Lombardy, and J.~Sakarovitch.
\newblock On the number of broken derived terms of a rational expression.
\newblock \emph{Journal of Automata, Languages and Combinatorics}, 15\penalty0
  (1/2):\penalty0 27--51, 2010.

\bibitem[Antimirov(1996)]{antimirov.1996.tcs}
V.~Antimirov.
\newblock Partial derivatives of regular expressions and finite automaton
  constructions.
\newblock \emph{TCS}, 155\penalty0 (2):\penalty0 291--319, 1996.

\bibitem[Brzozowski(1964)]{brzozowski.64.jacm}
J.~A. Brzozowski.
\newblock Derivatives of regular expressions.
\newblock \emph{J. ACM}, 11\penalty0 (4):\penalty0 481--494, 1964.

\bibitem[Caron et~al.(2011)Caron, Champarnaud, and Mignot]{caron.2011.lata.2}
P.~Caron, J.-M. Champarnaud, and L.~Mignot.
\newblock Partial derivatives of an extended regular expression.
\newblock In \emph{LATA}, vol. 6638 of \emph{LNCS}, pp. 179--191. Springer,
  2011.

\bibitem[Champarnaud et~al.(2007)Champarnaud, Ouardi, and
  Ziadi]{champarnaud.2007.dlt.2}
J.-M. Champarnaud, F.~Ouardi, and D.~Ziadi.
\newblock An efficient computation of the equation $\mathbb{K}$-automaton of a
  regular $\mathbb{K}$-expression.
\newblock In \emph{DLT}, vol. 4588 of \emph{LNCS}. Springer, 2007.

\bibitem[Demaille(2016)]{demaille.16.arxiv}
A.~Demaille.
\newblock Derived-term automata for extended weighted rational expressions.
\newblock Technical Report 1605.01530, arXiv, May 2016.
\newblock URL \url{http://arxiv.org/abs/1605.01530}.

\bibitem[Demaille et~al.(2013)Demaille, Duret-Lutz, Lombardy, and
  Sakarovitch]{demaille.13.ciaa.2}
A.~Demaille, A.~Duret-Lutz, S.~Lombardy, and J.~Sakarovitch.
\newblock Implementation concepts in {V}aucanson 2.
\newblock In \emph{CIAA'13}, vol. 7982 of \emph{LNCS}, pp. 122--133, July 2013.

\bibitem[Kaplan and Kay(1994)]{kaplan.1994.cl}
R.~M. Kaplan and M.~Kay.
\newblock Regular models of phonological rule systems.
\newblock \emph{Comput. Linguist.}, 20\penalty0 (3):\penalty0 331--378, Sept.
  1994.

\bibitem[Lombardy and Sakarovitch(2005)]{lombardy.2005.tcs}
S.~Lombardy and J.~Sakarovitch.
\newblock Derivatives of rational expressions with multiplicity.
\newblock \emph{TCS}, 332\penalty0 (1-3):\penalty0 141--177, 2005.

\bibitem[Makarevskii and Stotskaya(1969)]{makarevskii.1969.csa}
A.~Y. Makarevskii and E.~D. Stotskaya.
\newblock Representability in deterministic multi-tape automata.
\newblock \emph{Cybernetics and System Analysis}, 5\penalty0 (4):\penalty0
  390--399, 1969.

\bibitem[Owens et~al.(2009)Owens, Reppy, and Turon]{owens.2009.jfp}
S.~Owens, J.~Reppy, and A.~Turon.
\newblock Regular-expression derivatives re-examined.
\newblock \emph{J. Funct. Program.}, 19\penalty0 (2):\penalty0 173--190, Mar.
  2009.

\bibitem[Rutten(2003)]{rutten.2003.tcs}
J.~J. M.~M. Rutten.
\newblock Behavioural differential equations: a coinductive calculus of
  streams, automata, and power series.
\newblock \emph{TCS}, 308\penalty0 (1-3):\penalty0 1--53, 2003.

\bibitem[Sakarovitch(2009)]{sakarovitch.09.eat}
J.~Sakarovitch.
\newblock \emph{Elements of Automata Theory}.
\newblock Cambridge University Press, 2009.
\newblock Corrected English translation of \emph{\'El\'ements de th\'eorie des
  automates}, Vuibert, 2003.

\bibitem[Thompson(1968)]{thompson.68.cacm}
K.~Thompson.
\newblock Programming techniques: Regular expression search algorithm.
\newblock \emph{Commun. ACM}, 11\penalty0 (6):\penalty0 419--422, 1968.

\end{thebibliography}

\begin{longenv}
\appendix
\section{Appendix}

\subsection{Proof of \cref{lem:poly:ops}}
\label{proof:lem:poly:ops}
\begin{proof}[\cref{lem:poly:ops}]
  The first three equations are straightforward to prove.
  \begin{align*}
   \sem{\Pd \tuple \Qd}
    &= \semBig{\bigoplus_{(i,j) \in [n]\times [m]} \Lmul{k_i\cdot h_j}{\Ed_i \tuple \Fd_j}}
    \\
    &= \sum_{(i,j) \in [n]\times[m]} \lmul{k_i\cdot h_j}{\sem{\Ed_i \tuple \Fd_j}}
    \\
    &= \sum_{(i,j) \in [n]\times[m]} \lmul{k_i\cdot h_j}{\sem{\Ed_i} \tuple \sem{\Fd_j}}
    \\
    &= \parenBig{\sum_{i \in [n]} \lmul{k_i}{\sem{\Ed_i}}}
      \tuple \parenBig{\sum_{j \in [m]} \lmul{h_j}{\sem{\Fd_j}}}
    \\
    &= \semBig{\bigoplus_{i \in [n]} \Lmul{k_i}{\Ed_i}}
      \tuple \semBig{\bigoplus_{j \in [m]} \Lmul{h_j}{\Fd_j}}
    \\
    &= \sem{\Pd} \tuple \sem{\Qd} &\hfill \qed
  \end{align*}
\end{proof}

\subsection{Derived-Term Algorithm}
\label{sec:algo}

\begin{algorithm}[H]
  \SetKwInOut{Input}{Input}
  \SetKwInOut{Output}{Output}

  \Input{$\Ed$, a rational expression}
  \Output{$\bra{E, I, T}$ an automaton (simplified notation)}
  \BlankLine

  $I(\Ed)$ := $\unK$ \tcp*{Unique initial state}
  $Q$ := Queue($\Ed$) \tcp*{A work list (queue) loaded with $\Ed$}
  \While{$Q$ is not empty}
  {
    $\Ed$ := pop($Q$)        \tcp*{A new state/expression to complete}
    $\Xd$ := $d(\Ed)$        \tcp*{The expansion of $\Ed$}
    $T(\Ed)$ := $\Xd(\eword)$  \tcp*{Final weight: the constant term}
    \ForEach(\tcp*[f]{For each first/polynomial in $\Xd$}){$a \odot[\Pd_a] \in \Xd$}
    {
      \ForEach(\tcp*[f]{For each monomial of $\Pd_a = \Xd(a)$}){$\Lmul{k}{\Fd} \in \Pd_a$}
      {
        $E(\Ed, a, \Fd)$ := $k$ \tcp*{New transition}
        \If{$\Fd \not\in Q$}
        {
          push($Q$, $\Fd$) \tcp*{$\Fd$ is a new state, to complete later}
        }
      }
    }
  }
\end{algorithm}

\subsection{Derived Terms}
\label{sec:dt}

We will prove that the states of $\Ac_\Ed$ are actually members of
$\TD(\Ed)$ (and $\Ed$ itself), a \emph{finite} set of expressions, called the
\dfn{derived terms} of $\Ed$.  $\TD(\Ed)$ admits a simple inductive
definition.

\begin{definition}[Derived Terms]
  The \dfn{true derived terms} of an expression $\Ed$ is $\TD(\Ed)$, the set
  of expressions defined inductively below:
  \begin{align*}
    \TD(\zed) &\coloneqq \emptyset \\
    \TD(\und) &\coloneqq \emptyset \\
    \TD(a)    &\coloneqq \{\und\} \ee \forall a \in A\\
    \TD(\Ed + \Fd) &\coloneqq \TD(\Ed) \cup \TD(\Fd) \\
    \TD(\lmul{k}{\Ed}) &\coloneqq \TD(\Ed) \ee \forall k \in \K \\
    \TD(\rmul{\Ed}{k}) &\coloneqq \{\rmul{\Ed_i}{k} \mid \Ed_i \in \TD(\Ed)\}  \ee \forall k \in \K \\
    \TD(\Ed\cdot \Fd) &\coloneqq  \{\Ed_i\cdot\Fd \mid \Ed_i \in \TD(\Ed)\} \cup \TD(\Fd) \\
    \TD(\Ed^*) &\coloneqq  \{\Ed_i\cdot\Ed^* \mid \Ed_i \in \TD(\Ed)\} \\
    \TD(\Ed\tuple \Fd)
    &  \coloneqq (\TD(\Ed) \tuple \TD(\Fd)) \cup (\{\und\}
      \tuple \TD(\Fd)) \cup (\TD(\Ed) \tuple \{\und\})
  \end{align*}

  The \dfn{derived terms} of an expression $\Ed$ is
  $\D(\Ed) \coloneqq \TD(\Ed) \cup \{\Ed\}$.
\end{definition}

\begin{lemma}[Number of Derived Terms]
  For any $k$-tape expression $\Ed$,
  \begin{align*}
   \length{\TD(\Ed)} \le \prod_{i \in [k]}(\width{\Ed}_i + 1) \quad.
  \end{align*}
\end{lemma}

\begin{proof}
  It is simple to check by induction on $\Ed$ that for all cases, except
  tuple, $\TD(\Ed) \le \width{\Ed}$ (which is the classical result for
  single-tape expressions).  In the case of $\tuple$, it is clear that
  $\length{\TD(\Ed\tuple \Fd)} \le (\length{\TD(\Ed)} + 1) \cdot
  (\length{\TD(\Fd)} + 1)$, hence the result.
\end{proof}

\begin{lemma}[True Derived Terms and Single Expansion]
  \label{lem:tdtxpn:1}
  For any expression $\Ed$, $\exprs{d(\Ed)} \subseteq \TD(\Ed)$.
\end{lemma}

\begin{proof}
  Established by a simple verification of \cref{def:expa-of-expr}.  \qed
\end{proof}

The derived terms of derived terms of $\Ed$ are derived terms of $\Ed$.  In
other words, repeated expansions never ``escape'' the set of derived terms.

\begin{lemma}[True Derived Terms and Repeated Expansions]
  \label{lem:tdtxpn:*}
  Let $\Ed$ be an expression.  For all $\Fd \in \TD(\Ed)$,
  $\exprs{d(\Fd)} \subseteq \TD(\Ed)$.
\end{lemma}

\begin{proof}
  This will be proved by induction over $\Ed$.
  \begin{description}
  \item[Case $\Ed = \zed$ or $\Ed = \und$.]  Impossible, as then $\TD(\Ed) =
    \emptyset$.

  \item[Case $\Ed = a$.] Then $\TD(\Ed) = \{\und\}$, hence $\Fd = \und$ and
    therefore $d(\Fd) = d(\und) = \bra{\zeK}$, so
    $\exprs{d(\Fd)} = \emptyset \subseteq \TD(\Ed)$.

  \item[Case $\Ed = \Gd + \Hd$.]  Then $\TD(\Ed) = \TD(\Gd) \cup \TD(\Hd)$.
    Suppose, without loss of generality, that $\Fd \in \TD(\Gd)$.  Then, by
    induction hypothesis,
    $\exprs{d(\Fd)} \subseteq \TD(\Gd) \subseteq \TD(\Ed)$.

  \item[Case $\Ed = \lmul{k}{\Gd}$.]  Then if
    $\Fd \in \TD(\lmul{k}{\Gd}) = \TD(\Gd)$, so by induction hypothesis
    $\exprs{d(\Fd)} \subseteq \TD(\Gd) = \TD(\lmul{k}{\Gd}) = \TD(\Ed)$.

  \item[Case $\Ed = \rmul{\Gd}{k}$.]  Then
    $\forall \Fd \in \TD(\rmul{\Gd}{k}) = \{\rmul{\Gd_i}{k} \mid \Gd_i \in
    \TD(\Gd)\}$,
    there exists an $i$ such that $\Fd = \rmul{\Gd_i}{k}$.  Then
    $d(\Fd) = d(\rmul{\Gd_i}{k}) = \rmul{d(\Gd_i)}{k}$ hence
    $\exprs{d(\Fd)} = \exprs{\rmul{d(\Gd_i)}{k}}$.

    Since $\Gd_i \in \TD(\Gd)$, by induction hypothesis
    $\exprs{d(\Gd_i)} \subseteq \TD(\Gd)$, so by definition of the right
    exterior product of expansions (and polynomials),
    $\exprs{\rmul{d(\Gd_i)}{k}} \subseteq \TD(\rmul{\Gd}{k}) = \TD(\Ed)$.

    Hence $\exprs{d(\Fd)} \subseteq \TD(\Ed)$.

  \item[Case $\Ed = \Gd \cdot \Hd$.]  Then
    $\TD(\Ed) = \{\Gd_i\cdot\Hd \mid \Gd_i \in \TD(\Gd)\} \cup \TD(\Hd)$.
    \begin{itemize}
    \item If $\Fd = \Gd_i\cdot\Hd$ with $\Gd_i \in \TD(\Gd)$, then
      $d(\Fd) = d(\Gd_i\cdot\Hd) = \ep(\Gd_i)\cdot \Hd \oplus
      \lmul{\ec(\Gd_i)}{d(\Hd)}$.

      Since $\Gd_i \in \TD(\Gd)$ by induction hypothesis
      $\exprs{\ep(\Gd_i)} = \exprs{d(\Gd_i)} \subseteq \TD(\Gd)$.  By
      definition of the product of an expansion by an expression,
      $\exprs{\ep(\Gd_i)\cdot \Hd} \subseteq \{\Gd_j \cdot \Hd \mid \Gd_j
      \in \TD(\Gd)\} \subseteq \TD(\Gd \cdot \Hd) = \TD(\Ed)$.

    \item If $\Fd \in \TD(\Hd)$, then by induction hypothesis $\exprs{d(\Fd)}
      \subseteq \TD(\Hd) \subseteq \TD(\Ed)$.
    \end{itemize}

  \item[Case $\Ed = \Gd^*$.]  If
    $\Fd \in \TD(\Ed) = \{\Gd_i\cdot\Gd^* \mid \Gd_i \in \TD(\Gd)\}$, i.e., if
    $\Fd = \Gd_i\cdot\Gd^*$ with $\Gd_i \in \TD(\Gd)$, then
    $d(\Fd) = d(\Gd_i\cdot\Gd^*) = \ep(\Gd_i)\cdot \Gd^* \oplus
    \lmul{\ec(\Gd_i)}{d(\Gd^*)}$,
    so
    $\exprs{d(\Fd)} \subseteq \exprs{\ep(\Gd_i)\cdot \Gd^*} \cup
    \exprs{d(\Gd^*)}$.\footnote{Given
      two expansions $\Xd_1, \Xd_2$,
      $\exprs{\Xd_1 \oplus \Xd_2} \subseteq \exprs{\Xd_1} \cup
      \exprs{\Xd_2}$,
      but they may be different; consider for instance
      $\Xd_1 = a \odot[\Lmul{1}{\und}]$ and
      $\Xd_2 = a \odot[\Lmul{-1}{\und}]$ with $\K = \Z$.}  We will show that
    both are subsets of $\TD(\Ed)$, which will prove the result.

    Since $\Gd_i \in \TD(\Gd)$, by induction hypothesis,
    $\exprs{\ep(\Gd_i)} = \exprs{d(\Gd_i)} \subseteq \TD(\Gd)$, so by
    definition of a product of an expansion by an expression,
    $\exprs{\ep(\Gd_i) \cdot \Gd^*} \subseteq \{\Gd_j \cdot \Gd_j^* \mid
    \Gd_j \in \TD(\Gd)\} = \TD(\Ed)$.

    By \cref{lem:tdtxpn:1} $\exprs{d(\Gd^*)} \subseteq \TD(\Gd^*) = \TD(\Ed)$.

  \item[Case $\Ed = \Gd \tuple \Hd$.]%
    Let $\Fd \in \TD(\Ed) = \TD(\Gd) \tuple \TD(\Hd)$, i.e., let
    $\Fd = \Gd_i\tuple\Hd_j$ with $\Gd_i \in \TD(\Gd), \Hd_j \in \TD(\Hd)$,
    then by induction hypothesis $\exprs{d(\Gd_i)} \subseteq \TD(\Gd)$ and
    $\exprs{d(\Hd_j)} \subseteq \TD(\Hd)$.  So, by definition of the tupling
    of expansions
    $\exprs{d(\Gd_i)\tuple d(\Hd_j)} \subseteq \TD(\Gd)\tuple \TD(\Hd) =
    \TD(\Ed)$.

    We have $d(\Fd) = d(\Gd_i\tuple\Hd_j) = d(\Gd_i) \tuple d(\Hd_j)$, so
    $\exprs{d(\Fd)} = \exprs{d(\Gd_i) \tuple d(\Hd_j)} \subseteq \TD(\Ed)$.
    \qed
  \end{description}
\end{proof}

\begin{lemma}[Derived Terms and Repeated Expansions]
  Let $\Ed$ be an expression.  For all $\Fd \in \D(\Ed)$,
  $\exprs{d(\Fd)} \subseteq \TD(\Ed)$.
\end{lemma}
\begin{proof}
  Since $\D(\Ed) = \TD(\Ed) \cup \{\Ed\}$, this is an immediate consequence
  of \cref{lem:tdtxpn:1,lem:tdtxpn:*}.
\end{proof}

\subsection{Multitape Derivatives}
\label{sec:derivatives}

We reproduce here the definition of constant terms and derivatives from
Lombardy et al \citep[p.~148 and Def.~2]{lombardy.2005.tcs}, with our
notations and covering multitape expressions.  To facilitate reading,
weights such as the constant term are written in angle brackets, although so
far this was reserved to syntactic constructs.

\begin{definition}[Constant Term and Derivative]
  \label{def:ctder}
  \begin{align}
    \label{eq:der:cst}
    c(\zed) &\coloneqq \bra{\zeK}, &
    \da{\zed} &\coloneqq \zed,
    \\
    \notag
    c(\und) & \coloneqq \bra{\unK}, &
    \da{\und}& \coloneqq \zed,
    \\
    \label{eq:der:label}
    c(a) &\coloneqq \bra{\zeK}, \forall a \in A, &
    \da{b} &\coloneqq
      \und \text{ if $b = a$, }
      \zed \text{ otherwise,}
    \displaybreak[0]
    \\
    \label{eq:der:add}
    c(\Ed+\Fd) &\coloneqq c(\Ed) + c(\Fd), &
    \da{(\Ed + \Fd)} &\coloneqq \da{\Ed} \oplus \da{\Fd},
    \displaybreak[0]
    \\
    \label{eq:der:lmul}
    c(\lmul{k}{\Ed}) &\coloneqq \lmul{k}{c(\Ed)}, &
    \da{(\lmul{k}{\Ed})} &\coloneqq \lmul{k}{\left(\da{\Ed}\right)},
    \displaybreak[0]
    \\
    \label{eq:der:mul}
    c(\Ed \cdot \Fd) &\coloneqq c(\Ed) \cdot c(\Fd), &
    \da{(\Ed \cdot \Fd)} &\coloneqq \left(\da{\Ed}\right)\cdot \Fd \oplus \lmul{c(\Ed)}{\da{\Fd}},
    \\
    \label{eq:der:star}
    c(\Ed^*) &\coloneqq c(\Ed)^*,&
    \da{\Ed^*} &\coloneqq \lmul{c(\Ed)^*}{\left(\da{\Ed}\right)\cdot \Ed^*},
    \\
    \label{eq:der:tuple}
    c(\Ed \tuple \Fd)     & \coloneqq c(\Ed) \cdot c(\Fd),
    & \derivative{a|b}{(\Ed \tuple \Fd)},
    & \coloneqq \da{\Ed}\tuple\db{\Fd},\\ \notag
    && \derivative{a|\eword}{(\Ed \tuple \Fd)},
    & \coloneqq \lmul{c(\Fd)}{(\da{\Ed} \tuple \und)}, \\ \notag
    &&\derivative{\eword|b}{(\Ed \tuple \Fd)},
    & \coloneqq \lmul{c(\Ed)}{(\und \tuple \db{\Fd})}.
  \end{align}
  where \cref{eq:der:star} applies iff $c(\Ed)^*$ is defined in $\K$.
\end{definition}

\begin{lemma}
  \label{prop:expa:der}
  For any expression $\Ed$, $d(\Ed)(\eword) = c(\Ed)$, and
  $d(\Ed)(a) = \da{\Ed}$.
\end{lemma}

\begin{proof}
  A straightforward induction on $\Ed$.  The cases of constants and letters
  are immediate consequences of \cref{eq:der:cst,eq:der:label} on the one
  hand, and \cref{eq:epn:cst} on the other hand.  Equation \Cref{eq:epn:add}
  matches \cref{eq:der:add,eq:der:lmul}.  Multiplication (concatenation) is
  again barely a change of notation between \cref{eq:epn:mul} and
  \cref{eq:der:mul}, and likewise for the Kleene star
  (\cref{eq:epn:star,eq:der:star}) and tuple
  (\cref{eq:epn:tuple,eq:der:tuple}, using \cref{eq:epn:tuple:epn}).  \qed
\end{proof}

Note that, if we were to define the derivative with respect to the empty
word as the constant term, i.e.,
$\derivative{\eword}{\Ed} \coloneqq c(\Ed)$, then the previous definition
would simplify, for some operators, to:
\newcommand{\dl}[1]{\derivative{\ell}{#1}}
\begin{align*}
  \dl{(\Ed+\Fd)} &\coloneqq \dl{\Ed} + \dl{\Fd},
  \\
  \dl{(\lmul{k}{\Ed})} &\coloneqq \lmul{k}{(\dl{\Ed})},
  \\
  \derivative{\ell\tuple\ell'}(\Ed \tuple \Fd)  & \coloneqq \derivative{\ell}(\Ed) \tuple \derivative{\ell'}(\Fd).
\end{align*}
where for any weights $k, k', k \tuple k' \coloneqq k \cdot k'$.
\end{longenv}
\end{document}

